%% file: coverage_arxiv.tex
\documentclass[11pt]{article}
\usepackage[dvipsnames,usenames]{color}
\usepackage{amsfonts,amsmath,amssymb,amsthm,mathtools}
\usepackage{amsfonts,amsmath,amssymb,mathtools}
\usepackage{algorithmic,algorithm}
\usepackage{times}
\usepackage{bm}
\usepackage{xspace}
\usepackage{url}
\usepackage[pagebackref,letterpaper=true,colorlinks=true,pdfpagemode=none,urlcolor=blue,linkcolor=blue,citecolor=BrickRed,pdfstartview=FitH]{hyperref}
\usepackage{xspace,prettyref}

\let\pref=\prettyref
\newcommand{\savehyperref}[2]{\texorpdfstring{\hyperref[#1]{#2}}{#2}}

\newcommand{\ignore}[1]{}

\newrefformat{eq}{\savehyperref{#1}{\textup{(\ref*{#1})}}}
\newrefformat{lem}{\savehyperref{#1}{Lemma~\ref*{#1}}}
\newrefformat{def}{\savehyperref{#1}{Definition~\ref*{#1}}}
\newrefformat{thm}{\savehyperref{#1}{Theorem~\ref*{#1}}}
\newrefformat{cor}{\savehyperref{#1}{Corollary~\ref*{#1}}}
\newrefformat{cha}{\savehyperref{#1}{Chapter~\ref*{#1}}}
\newrefformat{sec}{\savehyperref{#1}{Section~\ref*{#1}}}
\newrefformat{app}{\savehyperref{#1}{Appendix~\ref*{#1}}}
\newrefformat{tab}{\savehyperref{#1}{Table~\ref*{#1}}}
\newrefformat{fig}{\savehyperref{#1}{Figure~\ref*{#1}}}
\newrefformat{hyp}{\savehyperref{#1}{Hypothesis~\ref*{#1}}}
\newrefformat{alg}{\savehyperref{#1}{Algorithm~\ref*{#1}}}
\newrefformat{item}{\savehyperref{#1}{Item~\ref*{#1}}}
\newrefformat{step}{\savehyperref{#1}{step~\ref*{#1}}}
\newrefformat{conj}{\savehyperref{#1}{Conjecture~\ref*{#1}}}
\newrefformat{fact}{\savehyperref{#1}{Fact~\ref*{#1}}}
\newrefformat{prop}{\savehyperref{#1}{Proposition~\ref*{#1}}}
\newrefformat{claim}{\savehyperref{#1}{Claim~\ref*{#1}}}

\newtheorem{theorem}{Theorem}[section]
\newtheorem{lemma}[theorem]{Lemma}
\newtheorem{claim}[theorem]{Claim}

\newtheorem{corollary}[theorem]{Corollary}
\newtheorem{fact}{Fact}
\newtheorem{conjecture}[theorem]{Conjecture}

\newtheorem{definition}{Definition}
\newtheorem{remark}{Remark}

\renewcommand{\ignore}[1]{}

\renewcommand{\Pr}{\mathop{\bf Pr\/}}

\newcommand{\R}{\mathbb R}

\newcommand{\F}{\mathbb F}

\newcommand{\eps}{\epsilon}

\newcommand{\poly}{\mathrm{poly}}

\newcommand{\wh}{\widehat}

\newcommand{\calP}{{\cal P}}

\newcommand{\calD}{{\cal D}}

\newcommand{\calJ}{{\cal J}}

\newcommand{\bx}{\boldsymbol{x}}

\newcommand{\abs}[1]{\left\lvert #1 \right\rvert}

\usepackage{fullpage, prettyref}

\def\eps{\varepsilon}
\def\bar{\overline}

\def\Pr{\mathbf{P}}
\def\Exp{\mathbf{E}}

\def\chi{X}
\def\F{{\mathcal F}}
\def\m{{[m]}}
        {\hspace*{\fill}$\Box$\par}

\begin{document}

\begin{titlepage}
    \thispagestyle{empty}

\title{Testing Coverage Functions}


\author{ Deeparnab Chakrabarty\\ \small Microsoft Research, Bangalore \\ \small Email: {\tt dechakr@microsoft.com} \and Zhiyi Huang \\  \small University of Pennsylvania \\  \small Email: {\tt hzhiyi@cis.upenn.edu}} \date{} 
\maketitle

\begin{abstract}
    \thispagestyle{empty}
A {\em coverage function} $f$ over a ground set $[m]$ is associated with a universe $U$ of weighted elements and $m$ sets $A_1,\ldots,A_m \subseteq U$, and for any $T\subseteq [m]$, $f(T)$ is defined as the total weight of the elements in the union $\cup_{j\in T} A_j$. Coverage functions are an important special case of submodular functions, and arise in many applications, for instance as a class of utility functions of agents in combinatorial auctions.

Set functions such as coverage functions often lack succinct representations, and in algorithmic applications, an access to a value oracle is assumed. In this paper, we ask whether one can test if a given oracle is that of a coverage function or not. We demonstrate an algorithm which makes $O(m|U|)$ queries to an oracle of a coverage function and completely reconstructs it. This gives a polytime tester for {\em succinct} coverage functions for which $|U|$ is polynomially bounded in $m$. In contrast, we demonstrate a set function which is ``far" from coverage, but requires $2^{\tilde{\Theta}(m)}$ queries to distinguish it from the class of coverage functions.
\end{abstract}
\end{titlepage}

\input{intro}

\input{succinct}

\input{general}

\input{distances}

\noindent
{\bf Acknowledgements.}
The authors wish to thank C. Seshadhri, Jan Vondr\'ak , Sampath Kannan, Jim Geelen and Mike Saks
for very fruitful conversations. DC especially thanks Sesh for illuminating conversations over the
past few years, and Mike for asking insightful questions.


%

\bibliographystyle{splncs03}
\bibliography{coverage}

\end{document}

%% file: intro.tex
\section{Introduction} \label{sec:intro}

Submodular set functions are set functions $f:2^\m \mapsto \R$ defined
over a ground set $[m]$ which satisfy the property: $f(S\cap T) +
f(S\cup T) \le f(S) + f(T)$.  These are arguably the most extensively
studied set functions, and arise in various fields such as combinatorial
optimization, computer science, electrical engineering, economics, etc.
In this paper, we focus on a particular class of submodular functions,
called coverage functions.

Coverage functions arise out of families of sets over a universe.
Given a universe $U$ and sets $A_1,\cdots,A_m \subseteq U$, the coverage of a collection of sets $T\subseteq [m]$
is the number of elements in the union $\bigcup_{j\in T} A_j$. More generally, each element $i\in U$ has a weight $w_i \ge 0$, inducing the  function $f:2^{[m]} \mapsto \R_{\ge 0}$:
$$\textstyle\forall T \subseteq \m : f(T) = w\left(\bigcup_{i \in T}
A_i\right)$$
with the usual notation of $w(S) := \sum_{i\in S} w_i$. A set function is called a {\em coverage function} iff~$f$ 
is induced by a set system as described above. 
 In the definition above, the size of the universe $U$ of the inducing set system can be arbitrarily large. 
We call a coverage function {\em succinct} if $|U|$ is bounded by a fixed polynomial in $m$. 

Coverage functions arise in many applications (plant location \cite{CFN}, machine learning \cite{Krause}); an important one being that in combinatorial auctions \cite{LLN01,Nis-Chapter}. Utilities of agents are often modeled as coverage functions -- agents are thought to have certain requirements (the universe $U$) and the items being auctioned (the $A_i$'s) fulfill certain subsets of these. Many auction mechanisms take advantage of the specific property of these utility functions; a notable one is the recent work of 
Dughmi,  Roughgarden and Yan \cite{DRY11} who give $O(1)$-approximate truthful mechanisms when utilities of agents are coverage. (Such a result is not expected for general submodular functions \cite{DV11}.) 
%
%

In general, set functions have exponentially large (in $m$) description, and algorithmic applications often assume access to a value oracle which returns $f(T)$ on being queried a subset $T\subseteq \m$. Efficient algorithms making only polynomially many queries to this oracle, exploit the coverage property of the underlying function to ensure correctness.
This raises the question we address in this paper: 
\begin{quote}{\em Can one test, in polynomial time, whether the oracle at hand is indeed that of a coverage (or a `close' to coverage) function?}
\end{quote}
\noindent It is easy to see that the parenthesized qualification in the above question is necessary. Using property testing parlance  \cite{GGR,Gol-surv}, we say a function is $\eps$-far from coverage if it needs to be modified in $\eps$-fraction of the points to make it a coverage function.

Our first result (\pref{thm:learn-cov}) is a reconstruction algorithm
which makes $O(m|U|)$ queries to a value oracle of a {\em true} coverage
function and reconstructs the coverage function, that is, deduces the
underlying set system $(U;A_1,\ldots,A_m)$ and weights of the elements
in $U$. Such an algorithm can be used distinguish coverage functions
with those which are $\eps$-far from being coverage (\pref{cor:cor1}).
In particular, for succinct coverage functions, the answer to the above
question is yes.

Our second result illustrates why the testing question may have a negative answer for general coverage functions. We show that certifying `non-coverageness' requires exponentially many queries. To explain this, let us first consider a certificate of a non-submodularity. By definition, for any non-submodular function $f$, there must exist sets $S, T, S\cup T,$ and $S\cap T$ such that $f(S) + f(T) < f(S\cup T) + f(S\cap T)$. Therefore, four queries (albeit non-deterministic) to a value oracle of $f$ certifies non-submodularity of $f$. In contrast, we exhibit non-coverage functions for which {\em any} certificate needs to query the function at exponentially many sets (\pref{cor:cert}).

In fact, from just the definition of coverage functions it is not {\em a priori} clear what a certificate for coverageness should be.
In \pref{sec:bipbase}, we show that a particular linear transformation (the $W$-transform) of set functions can be used: we show a function $f$ is coverage iff all its $W$-coefficients are non-negative. 
This motivates a new notion of  distance to coverageness which we call  $W$-distance: a set function has $W$-distance $\eps$ if at least an $\eps$-fraction of the $W$-coefficients are negative. This notion of distance captures the density of certificates to non-coverageness. Our lower bound results show that testing coverage functions against this notion of distance is infeasible: we construct set functions with $W$-distance at least ~$1 - e^{- \Theta(m)}$ which require $2^{\Theta(m)}$ queries to distinguish them from coverage functions (\pref{cor:wdist}). 

How is the usual notion of distance to coverage related to the $W$-distance? We show in \pref{sec:distances} that there are functions which are far in one notion but close in the other. Nonetheless, we believe that the functions we construct for our lower bounds also have large (usual) distance to coverage functions. We prove this assuming a conjecture on the number of roots of certain multilinear polynomials; we also provide some partial evidence for this conjecture.

\paragraph{Related Work}

The work most relevant to, and indeed which inspired this paper, is that
by Seshadhri and Vondr\'ak \cite{SV11}, where the authors address the
question of testing general submodular set functions. The authors focus
on a particular simple testing algorithm, the ``square tester'', which
samples a random set $R$, $i,j\notin R$ and checks whether or not 
$ f(R,i,j) + f(R) \le f(R,i) + f(R,j)$. \cite{SV11} show that $\eps^{-\tilde{O}(\sqrt{m})}$ random samples are sufficient to distinguish submodular functions from those $\eps$-far from submodularity, and furthermore, at least $\eps^{- 4.8}$ samples are necessary. Apart from the obvious problem of closing this rather large gap, the authors of \cite{SV11} suggest tackling special, well-motivated cases of submodularity. In fact, the question of testing coverage functions was specifically raised by Seshadhri in \cite{Sesh-Open} (attributed to N. Nisan).

It is instructive to compare our results with that of \cite{SV11}. Firstly, although coverage functions are a special case of submodular functions, the sub-exponential time tester of \cite{SV11} {\em does not} imply a tester for coverage functions. This is because a function might be submodular but far from coverage; in fact, the function $f^*$ in our lower bound result is submodular. Given our result that there are no small certificates of non-coverageness, we believe testing coverageness is {\em harder} than testing submodularity.

A recent relevant paper is that of Badanidiyuru et. al. \cite{BDF+12}. Among other results, \cite{BDF+12} shows that any coverage function $f$ can be arbitrarily well approximated by a succinct coverage function. More precisely, if $f$ is 
defined via $(U;A_1,\ldots,A_m)$ with weights $w$, then for any $\eps>0$, there exists another coverage function $f'$ defined via $(U';A'_1,\ldots,A'_m)$ with weights $w'$ such that $f'(T)$ is within $(1\pm \eps) f(T)$ such that $|U'| = \poly(m,1/\eps)$. This, in some sense shows that succinct coverage functions capture the essence of coverage functions.
Unfortunately, this `sketch' is found using random sampling on the
universe $U$ and it is open whether this can be obtained via polynomially
many queries to an oracle for $f$.
\ignore{
The question of testing submodularity was first raised by Parnas, Ron and Rubinfeld \cite{PRR03}, in a more generalized context than those of set functions, but gave results for the special case of testing whether a matrix is a Monge matrix or not. Their results are orthogonal to our work, and we refer the interested reader to their paper. Another strand of research which has originated in the recent years is that of {\em learning} submodular functions \cite{GHIM09,BH11,GHRU11}. In these works, oracle access to a submodular function is assumed, and the algorithm needs to `learn' enough via making polynomially many queries, so as to be able to evaluate the value of an unqueried set up to {\em multiplicative factors} which determines the quality of the tester. For instance, Goemans, Harvey, Iwata and Mirrokni \cite{GHIM09} show there exists a $\tilde{O}(\sqrt{m})$-factor learning algorithm,  and this is almost tight.
Note that for succinct coverage functions, in this terminology, our result implies we can perfectly learn them.
}

\subsection{The $W$-transform: Characterizing Coverage Functions}\label{sec:bipbase}

Given a set function  $f:2^{[m]} \mapsto \R_{\ge 0}$, we define the $W$-transform $w:2^{[m]}\setminus \emptyset \mapsto \R$ as

\begin{equation}\label{eq:maindef}
\forall S\in 2^\m\setminus\emptyset , \qquad w(S) = \sum_{T:S\cup T = \m} (-1)^{|S\cap T|+1}f(T)
\end{equation}

\noindent
We call the resulting set $\{w(S): S\subseteq [m]\}$ the $W$-coefficients of $f$. 
The $W$-coefficients are unique; this follows since the $(2^m-1)\times(2^m-1)$ matrix $M$ defined as $M(S,T) = (-1)^{|S\cap T| +1}$ if $S\cup T = [m]$ and $0$ otherwise, is full rank\footnote{One can check $M^{-1}(S,T) = 1$ if $S\cap T\neq \emptyset$ and $0$ otherwise. 
}.
Inverting we get the {\em unique} evaluation of $f$ in terms of its $W$-coefficients.

\begin{equation}\label{eq:wtof}
\forall T\subseteq [m], ~~~ f(T) = \sum_{S\subseteq \m: S\cap T\neq \emptyset} w(S)
\end{equation}

\noindent
If $f$ is a coverage function induced by the set system $(U;A_1,\cdots,A_m)$, then the function $w(S)$ precisely is the size of 
$\bigcap_{i\in S} A_i$ and is hence non-negative. This follows from the inclusion-exclusion principle.
Indeed the non-negativity of the $W$-coefficients is a characterization of coverage functions.
\begin{theorem}\label{thm:char-cov}
A set function $f:2^{[m]} \mapsto \R_{\ge 0}$ is coverage iff all its $W$-coefficients are non-negative.
\end{theorem}
\begin{proof}
Suppose that $f$ is a function with all $W$-coefficients non-negative. Consider a universe $U$ consisting of $\{S: S\subseteq [m]\}$ with weight of element $S$ being $w(S)$, the $S$th $W$-coefficient of $f$. Given $U$, for $i = 1\ldots m$, define $A_i  := \{S \subseteq [m]: i\in S\}$. For any $T\subseteq [m]$, 
$\bigcup_{i \in T} A_i = \{S\subseteq [m]: S\cap T \neq \emptyset\}$. From \eqref{eq:wtof} we get 
$f(T) = w\left(\bigcup_{i \in T} A_i\right)$ proving that $f$ is a coverage function. 

\indent
Suppose $f$ is a coverage function. By definition, there exists $(U; A_1, \ldots, A_m)$ with non-negative weights on elements in $U$ such that $f(T) = w\left(\bigcup_{i \in T} A_i\right)$. Each element in $S \in U$ corresponds to a subset of $[m]$ defined as $\{i: S \in A_i\}$. We may assume each element of $U$ corresponds to a unique subset; if more than one elements have the same incidence structure, we may merge them into one element with weight equalling sum of both the weights. This transformation doesn't change the function value and keeps the weights non-negative. Furthermore, we may also assume every subset on $[m]$ is an element of $U$ by giving weights equal to $0$; this doesn't change the function value either. In particular, $|U|$ may be assumed to be $2^m$. As before, one can check that for any $T\subseteq [m]$, 
$f(T) = \sum_{S: S\cap T\neq \emptyset} w(S)$. From \eqref{eq:wtof} we get that these are the $W$-coefficients of $f$, and are hence non-negative.
\end{proof}

\noindent
From the second part of the proof above, note that the positive $W$-coefficients of a coverage function $f$ correspond to the elements in the universe $U$. Let $\{S: w(S)> 0\}$ be the support of a coverage function $f$. Note that succinct coverage functions are precisely those with polynomial support size.

One can use \pref{thm:char-cov} to certify non-coverageness of a function $f$: one of its $W$-coefficients $w(S)$ must be negative, and  the function values in the summand of \eqref{eq:maindef} certifies it. Observe, however, that this certificate can be exponentially large. In \pref{sec:far} we'll show this is inherent in {\em any} certificate of coverageness.
The $W$-transformation also motivates the following notion of distance to coverage functions.

\begin{definition}
The $W$-distance of a function $f$ from coverage functions is the fraction of its  negative $W$-coefficients.
\end{definition}

\paragraph{Comparison with Fourier Transformation}
Readers who are familiar with the analysis of Boolean functions might
find \pref{eq:maindef} similar to the Fourier transformation. Indeed, if
we sum over all $T$ in the summation of \pref{eq:maindef} instead of
only over the $T$ s.t.~$S \cup T = \m$, then it becomes the Fourier
transformation. However, it is worth pointing out that due to this subtle
change, the $W$-transformation behaves quite differently to the
representation by Fourier basis. In particular, unlike the Fourier
basis, the basis of the $W$-transform is not orthonormal with
respect to the usual notion of inner product.

\def\x{{\bf x}}
\def\y{{\bf y}}
\def\F{{\mathcal F}}

%% file: succinct.tex
\section{Reconstructing Succinct Coverage Functions} \label{sec:succ}

Given a coverage function $f$, suppose 
$\{S_1,\ldots,S_n\}$ is the support of $f$. That is, these are the sets in the $W$-transform of $f$
with $w(S_i) > 0$, and all the other sets have weight $0$. We now give an algorithm to find these sets and weights using $O(mn)$ queries.  As a corollary, we will obtain a polynomial time algorithm for testing succinct coverage functions where $n = \mbox{poly}(m)$.

The procedure is iterative. 
The algorithm maintains a partition of $2^\m$ at all times, and for each part in the partition, stores the total weight of the all the sets contained in the part. We start of with the trivial partition containing all sets whose weight is given by $f([m])$.
In each iteration, these partitions are refined; for instance, in the first iteration we divide the partition into sets containing a given element $i$ and those that don't contain the element $i$. The total weights of the first collection can be found by querying $f(\{i\})$. Any time the sum of a part evaluates to $0$, we discard it and subdivide it no more\footnote{Familiar readers will observe the similarity of our algorithm and the Goldreich-Levin algorithm to compute `large' Fourier coefficients (see, for instance, \cite{ODonnell} for an exposition).}. After $m$ iterations, the remaining $n$ parts give the support sets and their weights. To describe formally, we introduce some notation.

Given a vector $\x \in \{0,1\}^k$ 
we associate a subset of $[k]$ containing the elements $i$ iff $\x(i) = 1$. At times, we abuse notation and use the vector to imply the subset.
Let $\F(\x) := \{S\subseteq \m: S\cap [k] = \x\}$, that is, subsets of $\m$ which ``match'' with the vector $\x$ on the first $k$ elements. Note that $|\F(\x)| = 2^{m-k}$, and $\{\F(\x):\x\in\{0,1\}^k\}$ is a partition of $2^\m$;
if $k=0$, then $\F(\x)$ is the trivial partition consisting of all subsets of $\m$.
Given $\x\in \{0,1\}^k$, we let $\x\oplus 0$ be the $(k+1)$ dimensional vector with $\x$ appended with a $0$. Similarly, define $\x\oplus 1$.
At the $k$th iteration, the algorithm maintains the partition $\{\F(\x):
\x\in \{0,1\}^k\}$ and the total weight of subsets in each $\F(\x)$. In
the subsequent iteration refines each partition $\F(\x)$ into
$\F(\x\oplus 0)$ and $\F(\x \oplus 1)$. However, if a certain weight of
a part of the partition evaluates to $0$, then the algorithm does not need to refine that part any further since all the weights of that subset {\em must} be zero. The algorithm terminates in $m$ iterations making  $O(mn)$ queries. 
We now give the refinement procedure. In what follows, we say a vector $\y \le \x$ if they are of the same dimension and $\y(i) = 1 \Rightarrow \x(i)=1$. We say $\y < \x$ if $\y\le\x$ and $\y\neq \x$.

\begin{figure}
    \centering
\fbox{\begin{minipage}{.9\textwidth}
    \noindent Procedure {\bf ~Refine}
    \begin{algorithmic}[1]
    	\STATE {\bf Input:} $0\le k\le m$, $\{w(\F(\x)) > 0: \x \in \{0,1\}^k\}$
	\STATE {\bf Output:} $\{w(\F(\x\oplus 0)), w(\F(\x\oplus 1)): \x\in\{0,1\}^k\}$ 
	\STATE Order $\{\x: w(\F(\x))>0 \}$ by increasing number of $1$'s breaking ties arbitrarily. \\
		     Call the order $\{\x_1,\ldots,\x_N\}$;
	\FOR{$i= 1 \to N$}
		 	\STATE Query $f([k]\setminus \x_i) = F^0_i$ and $f(([k]\setminus \x_i)\cup k+1) = F^1_i$.
			\STATE Define $\Delta_{\x_i} := F^1_i - F^0_i - \sum_{j < i} \Delta_{\x_j}$. \label{1}
			\STATE $w(\F(\x_i \oplus 1)) = \Delta_{\x_i}$; $w(\F(\x \oplus 0)) = w(\F(\x_i)) - \Delta_{\x_i}$
       \ENDFOR
     \end{algorithmic}
\end{minipage}}

\bigskip 

\fbox{\begin{minipage}{.9\textwidth}
    \noindent Procedure {\bf ~Recover Coverage}
    \begin{algorithmic}[1]
    	\STATE {\bf Input:} Value oracle to coverage function $f$,
	\STATE {\bf Output:} $\{S_1,\ldots,S_n\}$ with $w(S_i) > 0$.
	\STATE Initialize $k=0$, $\x$ to be the empty vector, and list $L$ to contain $\x$.\\ 
		     Let $w(\F(\x)) = f(\m)$.
	\FOR{$k = 1 \to m$} 
		\STATE Run {\bf Refine} on each $\x$ in list $L$ and remove it.
		\STATE Add $\x\oplus 0$ and $\x\oplus 1$ to $L$ only if the weights evaluate to positive.
	\ENDFOR
	\STATE For each $\x\in \{0,1\}^m$ in $L$, return corresponding set and weight calculated by the {\bf Refine} procedure.
	        \end{algorithmic}
\end{minipage}}
\end{figure}


\noindent

\begin{claim}
The procedure {\bf Refine} returns the correct weights of the refinement.
\end{claim}
\begin{proof}
It suffices to show that $\Delta_{\x_i} = w\left(\F(\x_i\oplus 1)\right) = \sum_{S:S\cap[k]=\x_i, k+1\in S} w(S)$.
The RHS equals
\begin{equation}\label{eq:3}
~\sum_{S:S\cap [k] \subseteq\x_i, ~k+1\in S} w(S) ~-~\sum_{\y < \x_i}~\sum_{S\cap[k] = \y, k+1\in S} w(S).
\end{equation}
\noindent
The first term above equates to 
$$\sum_{S:S\cap [k]\setminus \x_i = \emptyset, \atop k+1\in S} w(S) = \sum_{S:S\cap ([k]\setminus\x_i \cup k+1) \neq \emptyset} w(S) ~~- \sum_{S:S\cap ([k]\setminus\x_i)\neq \emptyset} w(S)~~= ~~F_i^1 - F_i^0$$
Note that the summation $\sum_{S\cap [k] = \y,k+1\in S} w(S)$ equals $0$
if $w(\F(\y)) = \sum_{S\cap[k]=\y}w(S)$ equals zero since $w(S)\ge 0$ for all $S$. Therefore, the second term in \eqref{eq:3} is precisely $\sum_{j<i}w(\F(\x_j\oplus 1))$.
If $i=1$, then this is $0$; for other $i$ this  equates to $\sum_{j< i}\Delta_{\x_j}$ by induction.

\end{proof}

\begin{theorem}\label{thm:learn-cov}
Given value oracle access to a coverage function $f$ with positive weight sets $\{S_1,\ldots,S_n\}$,
the procedure {\bf Recover Coverage} returns the correct weights with $O(mn)$ queries to the oracle.
\end{theorem}

\begin{proof}
Whenever a certain $w(\F(\x))$ evaluates to $0$, we can infer  that $w(S) = 0$ for all $S\in \F(\x)$ since
$f$ is a coverage function. It is also clear that the algorithm terminates in $m$ steps since the partition refines to singleton sets. The number of oracle accesses is proportional (twice) to the number of calls to the {\bf Refine} subroutine. The latter is at most $mn$ since in each iteration the number of parts remaining is at most the number of parts remaining in the end.
\end{proof}

\begin{corollary}\label{cor:cor1} 
    Given any $n$, there exists a $O(mn + \epsilon^{-1})$ time tester
    which will return {\sc Yes} for coverage functions having
    $W$-support size at most $n$, and return {\sc No} with
    $\Omega(1)$ probability for functions that are $\epsilon$-far from
    the set of coverage functions with $W$-support at most $n$.
\end{corollary}
\begin{proof}
Run the reconstruction algorithm described above. If we get a set with negative weight, return {\sc No}.  If we succeed, then if $f$ is truly a coverage function, we have derived the unique weights. 
We sample $O(\eps^{-1})$ random sets and compare the value of our computed function with that of the oracle; if the function is $\eps$-far from coverage, then we will catch it with probability $O(1)$.
\end{proof}




\begin{theorem} \label{thm:coveragelower}
   Reconstructing coverage functions on $m$ elements with $W$-support  size $n$
     requires at least $\Omega(mn/\log n)$ probes.
\end{theorem}

\begin{proof}
    Consider the bipartite graphs with $m$ and
    $n$ vertices on the $A$ and $U$ side.
     Let the weight be $1$ on all vertices in $U$.
    Each non-isomorphic (on permutation of the $U$ vertices) maps to a different coverage function over the $A$ side: the neighborhood of a vertex $A_i \in A$ is precisely the elements it contains.
    Note each such 
     graph corresponds to a way of allocating $n$
    identical balls ($U$-side vertices) into $2^m$ different bins 
    (different choice of set of adjacent $A$-side vertices). This number
    is at least
    ${2^m+n-1 \choose n-1} \ge \left(\frac{2^m}{n}\right)^{n-1}$.

    Hence, we need at least $\Omega(mn)$ bits of information. Notice
    that each probe of function value only provides $O(\log n)$ bits of
    information since the function value is always an integer between
    $0$ and $n$, we get the lower bound in \pref{thm:coveragelower}.
\end{proof}

%% file: general.tex
\section{Testing Coverage Functions is Hard?}\label{sec:far}

In this section we demonstrate a set function whose $W$-distance to
coverage functions is `large', but it takes exponentially many queries to
distinguish from coverage functions. In particular, the function has
$W$-coefficients $w(S) = -1$ if $|S| > k := k(m)$, and $w(S) = N$ if $|S| \le k$,
where $N$ is a positive integer and $k(m)$ is a growing function of $m$,
which will be precisely determined later. Let this function be called
$f^*$.

Firstly, observe that from \pref{eq:maindef} it follows that $w(S)$ can be precisely determined by querying the $2^{|S|}$ sets in $\{T: T\cup S=\m\} = \{\bar{S} \cup X: X\subseteq S\}$.
It follows that $f^*$ can be distinguished from coverage using $2^{k+1}$ queries. 

In this section we show an almost tight lower bound: Any tester which makes less than $2^k$ queries cannot distinguish $f^*$ from a coverage function. Our bound is information theoretic and holds even if the tester has infinite computation power. 
More precisely, we show that given the value of $f^*$ on a collection of sets $\calJ$ with $|\calJ| < 2^k$, there exists a {\em coverage} function $f$ which has the same values on the sets in $\calJ$. 
%

%

\begin{theorem}\label{thm:notester}
There exists a coverage function consistent with the queries of $f^*$ on $\calJ$ if $|\calJ| < 2^k$.
\end{theorem}

\begin{corollary}\label{cor:cert}
Any certificate of non-coverageness of $f^*$ must be of size at least $2^k$.
\end{corollary}
Setting $k(m) = m/4$, we get $f^*$ has $W$-distance at least $(1 - e^{-\Theta(m)})$, giving us:
\begin{corollary} \label{cor:wdist}
Any tester distinguishing between coverage functions and functions of $W$-distance as large as $(1 - e^{-\Theta(m)})$ needs at least $2^{\Theta(m)}$ queries.
\end{corollary}
\noindent
We give a sketch of the proof before diving into the details.
Suppose a tester queries the collection $\calJ$. We first observe that the existence of a coverage function consistent with the queries in $\calJ$ can be expressed as a set of linear inequalities. Using Farkas' lemma, we get a certificate of the {\em non-existence} of such a completion. 
This certificate, at a high level, corresponds to an assignment of values on the $m$-dimensional hypercube satisfying certain linear constraints. We show that if the parameter $N$ is properly chosen, most of these assignments can be assumed to be $0$. In the next step we use this property to show that unless the size of $|\calJ| \ge 2^k$, all the assignments need to be $0$ which contradicts the Farkas linear constraints, thereby proving the existence of the coverage function consistent with $\calJ$.

\subsection{Consistent Coverage Functions and Farkas Lemma}
Recall, from \pref{thm:char-cov}, a function $f:2^\m \mapsto \R_{\ge 0}$
is coverage iff it satisfies
\begin{align*}
    \forall S \subseteq \m: & &
    \textstyle
    \sum_{T : S \cup T = \m} (-1)^{\abs{S
    \cap T} + 1} f(T) & ~~\ge ~~0 \\
    \forall T \subseteq \m: & & f(T) & ~~\ge~~ 0
\end{align*}
Let $\calJ$ be the collection of sets on which the function $f^*$ has been queried.  Define 
$$\textstyle b(S) :=\sum_{T \in \calJ : S \cup T = \m} (-1)^{\abs{S \cap
T}} f^*(T) $$
\noindent
Therefore, if we can find assignments $f:2^\m\setminus \calJ \mapsto \R_{\ge 0}$ satisfying:
\begin{align}
    \forall S \subseteq \m: & & \textstyle \sum_{T \notin \calJ : S \cup
    T = \m} (-1)^{\abs{S \cap T} + 1} f(T) &~~ \ge ~~b(S) \label{eq:1} \\
    \forall T \notin \calJ: & & f(T) & ~~\ge ~~0 \label{eq:2}
\end{align}
\noindent
we can complete the queries on $\calJ$ to a coverage function. Applying
Farkas' lemma (see for instance \cite{BT}), we see that there is {\em
no} feasible solution to \pref{eq:1}, \pref{eq:2} if and only if there is a feasible solution 
$\alpha:2^\m \mapsto \R_{\ge 0}$ satisfying:
\begin{align}
    & & \textstyle \sum_{S \subseteq \m} \alpha(S) b(S) & ~~> ~~0
    \label{eq:f1} \\
    \forall T \notin \calJ : & & \textstyle \sum_{S : S \cup T = \m} 
    (-1)^{\abs{S \cap T} + 1} \alpha(S) & ~~\le ~~0 \label{eq:f2}\\
    \forall S \subseteq \m : & & \alpha(S) & ~~\ge ~~0 \label{eq:f3}
\end{align}

Now we define the parameter $N$ for the function $f^*$; let $N$ be any integer larger than $(2m)!$.
Note that this makes the values doubly exponential, but we are
interested in the power of an all powerful tester.  In the next lemma we
show that one can assume there is a feasible solution to \pref{eq:f1},
\pref{eq:f2}, and \pref{eq:f3} with half of the $\alpha(S)$'s set to $0$. 

\begin{lemma}\label{lem:mostalphaszero}
    If there exists $\alpha$ satisfying \pref{eq:f1}, \pref{eq:f2},
    and \pref{eq:f3}, then we may assume $\alpha_S = 0$ for all $S$
    such that $\abs{S} \le k$.
\end{lemma}
\noindent
Intuitively, what this lemma says is that the constraint \pref{eq:1} for 
sets of size $\le k$ should not help in 
catching the function not being coverage. This is because the true
function values satisfies the constraints with huge `redundancy': $\sum_{T
: S \cup T = \m} (-1)^{\abs{S \cap T} + 1} f^*(T) = N \gg 0$. Formally,
we can prove the lemma as follows.

\begin{proof}
    Suppose there is an $\alpha$ satisfying \pref{eq:f1},
    \pref{eq:f2}, and \pref{eq:f3}. Then, by scaling we may assume
    that 
    \begin{equation}
        \textstyle \sum_{S\subseteq \m} \alpha(S) = 1\label{eq:f4}
    \end{equation}
    Equivalently, there is a positive valued solution to the
    LP $\{\max \sum_{S\subseteq \m}b(S)\alpha(S):
    \pref{eq:f2}, \pref{eq:f3}, \pref{eq:f4}\}$. Choose $\alpha$ to be
    a basic feasible optimal solution. Such a solution makes $2^m$ of
    the inequalities in \pref{eq:f2}, \pref{eq:f3}, and \pref{eq:f4} tight, and therefore by Cramer's rule, each of the non-zero $\alpha(S) \ge \frac{1}{(2^m)!}$ since all coefficients are $\{-1,0,1\}$. 

    Now we show that if $\alpha$ is basic feasible and $N > (2^m)!$,
    then we must have that $\alpha(S) = 0$ for all $S$ such that $\abs{S}
    \le k$. We first note that $\forall S \subseteq \m$:
    $$w(S) = \sum_{T : S \cup T = \m}
    (-1)^{\abs{S \cap T} + 1} f^*(T) = \sum_{T \notin \calJ : S \cup T =
    \m} (-1)^{\abs{S \cap T} + 1} f^*(T) - b(S) \enspace.$$

Therefore, $\sum_{S \subseteq \m} \alpha(S) b(S) > 0$ and the above
equality imply that
$$\sum_{T \notin \cal J} \sum_{S : S \cup T =
        \m} \alpha(S) (-1)^{\abs{S \cap T} + 1} f^*(T) - \sum_{S \subseteq
        \m} \alpha(S) w_S = \sum_{S \subseteq \m} \alpha(S)
        b(S) > 0 \enspace.$$

    But by \pref{eq:f2}, $\sum_{S \subseteq \m} \alpha(S) (-1)^{\abs{S \cap T} +1} \le
    0$ for all $T \notin \cal J$, and $f^*(T) \ge 0$ for all $T \subseteq
    \m$. So we have that $\sum_{S \subseteq \m} \alpha(S) w(S) < 0$.
    Assume for contradiction that there exists $S_0$, $\abs{S_0} \le k$
    such that $\alpha_{S_0} \ne 0$.  From the earlier discussion we know
    that $\alpha_{S_0} \ge \frac{1}{(2^m)!}$. Therefore, we have 
    $\sum_{S \subseteq \m} \alpha(S) w(S) \ge \frac{1}{(2^m)!} N -
    \sum_{S \subseteq \m : \abs{S} \le k} \alpha(S) > 1 - 1 = 0$, a
    contradiction. The latter inequality follows from \pref{eq:f4} and
    our assumption that $N > (2^m)!$.
\end{proof}

\subsection{Nullity of Farkas Certificate}

In the following discussion, we assume without loss of generality $\alpha(S) = 0$ for all $S$, $\abs{S} \le k$. We will work with the following linear function of the $\alpha$'s. For a set $T$, define
$$\textstyle g(T) := \sum_{S : S \cup T = \m} (-1)^{\abs{S \cap T} +1}\alpha(S)$$
\noindent
From \pref{eq:f2}, we get $g(T) \leq 0$ for all $T\notin \calJ$. Inverting, we get 
\begin{equation}\label{eq:alpha-g}
    \textstyle
\alpha(S) = \sum_{T:T\cap S\neq \emptyset} g(T) ~=~ G - \sum_{T\subseteq\bar{S}} g(T), ~~~~ \textrm{where} ~G:= \sum_{T\subseteq\m}g(T)
\end{equation}
\noindent
We now show that if $\alpha(S) = 0$ for all $|S|\le k$, then $g(T)$ {\em must} be $>0$ for at least $2^k$ sets $T$. This will imply $|\calJ| \ge 2^k$. 
\begin{lemma}\label{lem:mostgtspositive}
If $\alpha(S) = 0$ for all $|S| \le k$, then $g(T) > 0$ for at least $2^k$ subsets $T\subseteq \m$.
\end{lemma}
\begin{proof}
Let $S^*$ be any minimal set with $\alpha(S^*) > 0$. Note that $|S^*| \ge k+1$.
From \pref{eq:alpha-g}, we get $\hat{G} := \alpha(S^*) = G -
\sum_{T\subseteq \bar{S^*}} g(T) > 0$.
Consider any $i\in S^*$. By minimality, we have $\alpha(S^*\setminus i) = 0$, giving us
$$\textstyle 0 = G - \sum_{T\subseteq \bar{S^*\setminus i}} g(T) = G - \sum_{T\subseteq\bar{S^*}} g(T) - \sum_{T\subseteq\bar{S^*}} g(T\cup i)$$ 
Therefore for all $i\in S^*$,
$\sum_{T\subseteq \bar{S^*}} g(T\cup i) = \hat{G} > 0$. 
\noindent
By induction, we can extend the above calculation to  any subset $X\subseteq S^*$,
\begin{equation}\label{eq:xx}
    \textstyle
\sum_{T\subseteq \bar{S^*}} g(T\cup X) = (-1)^{|X|+1}\hat{G}
\end{equation}
\noindent
Note that the summands in \pref{eq:xx} are disjoint for different sets $X$, and furthermore,
whenever $|X|$ is odd, the sum is $>0$ implying at least one of the summands must be positive for each odd subset $X\subseteq S^*$. This proves the lemma since $|S^*| = k+1$.
 \smallskip

\noindent
Proof of \pref{eq:xx}: Let's denote the sum $\sum_{T\subseteq \bar{S^*}} g(T\cup X)$ as $h(X)$.
So $\hat{G} = G  - h(\emptyset)$, and by induction, $h(Y) = (-1)^{|Y|+1}\hat{G}$ for every proper subset of $X$. Now, $\alpha(S^*\setminus X) = 0$ gives us
$$\textstyle 0 = G - \sum_{T\subseteq \bar{S^*\setminus X}} g(T) =  G - \sum_{Y\subseteq X} h(Y)$$
\noindent
Rearranging, 
$h(X) = G - \sum_{Y\subsetneq X} h(Y) = \hat{G} - \sum_{i=1}^{|X|-1} {|X|\choose i} (-1)^{i+1} \hat{G} = (-1)^{|X|+1}\hat{G}$
\end{proof}

\begin{proof}[\pref{thm:notester}]
    Suppose there is no consistent completion, implying $\alpha$'s
    satisfying  \pref{eq:f1}, \pref{eq:f2} and \pref{eq:f3}. By
    \pref{lem:mostalphaszero} and \pref{lem:mostgtspositive}, we get
    that if \pref{eq:f2} holds, then $|\calJ| \ge 2^k$.  
\end{proof}

%% file: distances.tex
\section{$W$-Distance and Usual Distance}
\label{sec:distances}

We first note that the two notions are unrelated; in particular, 
we show two functions each ``far'' in one notion, but
``near'' in the other. The  proofs of the following two lemmas
are provided in the following subsection.

\begin{lemma}\label{lem:WfarUnear}
There is a function with $W$-distance $1 - e^{-\Theta(m)}$ whose distance 
to coverage is $e^{-\Theta(m)}$.
\end{lemma}

\begin{lemma}\label{lem:WnearUfar}
There is a function with $W$-distance $O(m^2/2^m)$ whose distance to coverage is $\Omega(1)$.
\end{lemma}
%
\noindent
Despite the fact that the two notions are incomparable, we
argue that the lower bound example of  \pref{sec:far} is
in fact also far from coverage (with proper choice of $k(m)$) in the
usual notion of distance, under a reasonable conjecture about the
property of multilinear polynomials. Unfortunately, we are unable to
prove this conjecture and leave it as an open question.

\begin{conjecture} \label{conj:1}
For any $m$-variate multilinear polynomials $f(\bx) = \sum_{S
    \subseteq [m]} \lambda_S \prod_{i\in S} x_S$ with $\lambda_S < 0$ for all $|S| > k$,  has
    at most $O(k 2^m / \sqrt{m})$ zeroes on the hypercube $\{0,
    1\}^m$.
\end{conjecture}
\noindent
In fact, we conjecture that the maximum number of zeros is achieved when
the $k+1$ layers of function values in the ``middle of the hypercube'' are
zero, that is, $f(\bx) = 0$ iff.~$(m-k)/2 \le \|\bx\|_1 \le (m+k)/2$. 
At the end of this section, we  present
some evidence for this conjecture by giving a proving 
it for symmetric functions, that is, when $f(x_1,\ldots,x_m) = f(x_{\sigma(1)},\ldots,x_{\sigma(m)})$ 
for any permutation $\sigma$ of $\m$. We now show that the conjecture implies $f^*$ is far from coverage in 
the usual notion of distance.

\begin{lemma}
    Assuming \pref{conj:1}, with $k(m)=o(\sqrt{m})$, $f^*$ is $1-o(1)$ far from coverage.
\end{lemma}

\begin{remark}
    \pref{thm:notester} implies that
    $f^*$ requires superpolynomial queries to test as long as we have
    $k(m) = \omega(\log m)$.
\end{remark}

\begin{proof}
    Consider the coverage function $f'$ that is closest to $f^*$ in the usual notion of distance.
    Let $w',w^*$ be the
    $W$-coefficients of $f',f^*$. Define the function $\Delta f := f' -
    f^*$ and let $\Delta w := w' - w^*$. By linearity of $W$-transformation, we get that $\Delta w$ are the
    $W$-coefficients for $\Delta f$. Therefore,
    $$\textstyle \Delta f(T) = \sum_{S : T \cap S \ne \emptyset} \Delta
    w(S) = \sum_{S \subseteq [m]} \Delta w(S) (1 - \bm{1}_{T \cap S =
    \emptyset}) \enspace.$$
\noindent
Consider the following binary vector representation of $S\subseteq\m$: $\x\in \{0,1\}^m$ such that 
$\x_i = 0$ iff.~$i\in S$. Using this, the function $\Delta f$ can be interpreted as $\Delta f(\x) = W - \sum_{S\subseteq\m} w(S)\prod_{i\in S}x_i$. We are using here the fact that $T\cap S = \emptyset$ is equivalent to $S\subseteq \bar{T}$.
By our choice of $w^*$ and the assumption that $w'(S) \ge 0$ for all $S$, we have $\Delta w(S) \ge
    1$ for all $|S| > k$.  From \pref{conj:1}, we get that at
    most $O(k/\sqrt{m})$-fraction of the function values of $\Delta f$
    are zeroes. So $f'$ is at least $1 - O(k /\sqrt{m})$ far from $f^*$. The lemma follows since $k=o(\sqrt{m})$.
\end{proof}

\paragraph{Support for \pref{conj:1}: Proof for Symmetric Functions.}

Since $f$ is symmetric, each $\lambda_S$ is equal for sets of the same cardinality. Let
$\lambda_j$ denote the value of $\lambda_S$ when $|S| = j$. Then $f$ is equivalent to the 
function $g:\m\mapsto {\mathbb R}$
$$\textstyle g(i) = f(\bx : \|\bx\|_1 = i) = \sum_{j=0}^m \sum_{S : |S|
= j} \lambda_j \prod_{i\in S}x_i = \sum_{j=0}^m \lambda_j {i \choose j}
\enspace.$$ 

\noindent
By our assumption, $\lambda_j < 0$ for all $j > k$. Hence, all the
high order derivatives (at least $k+1$-th order) of $f$ are negative.
Intuitively, since the high order derivatives of $g$ are negative, there
are at most $k+1$ sign-changes of $g(i)$. Therefore, there are at most
$k+1$ different $i$'s such that $g(i) = 0$. This implies the conjecture for symmetric functions.

\subsection{Proof of \pref{lem:WfarUnear} and \pref{lem:WnearUfar}}

\begin{proof}[\pref{lem:WfarUnear}]
Let us consider a function $f$ which is similar to the lower bound
    example in \pref{sec:far}. More concisely, $f$'s
    $W$-representation satisfies that $w(S) = -1$ if $m > |S| > k$,
    and $w(S) = N$ if $|S| \le k$ or $S = [m]$. Here we will let $k =
    m/4$ and $N>0$ is a sufficiently large number; $N=5m!$ suffices.
        For simplicity, assume that $m$ is a multiple of $8$.

    First, it follows immediately from definition 
    that the fraction of weights that are negative is at least $(1 -
    e^{-\Theta(k)}) = (1 - e^{-\Theta(m)})$.
    Next, let us prove that there exists a coverage function $f'$ such
    that the function values of $f$ and $f'$ differ in at most
    $e^{-\Theta(m)}$ fraction of the entries. 
   Let $w'$ denote the
    $W$-representation of $f'$. Let $\Delta f = f' - f$ and $\Delta w
    = w' - w$. Note that $\Delta w$ is the $W$-representation of $\Delta f$.
    In the remainder, we will find a function $\Delta f$ over the subsets of $\m$ satisfying the following 
    properties: (a) $\Delta f$ is non-zero on at most a $e^{-\Theta(m)}$ fraction of the subsets, and 
    (b) the $W$-representation of $\Delta f$, that is $\Delta w$, has the property that 
    $\Delta w(S) \ge 1$ if $m > |S| > m/4$ and $\Delta w(S) \ge -N$ if $|S| \le k$ or $S = [m]$. 
   
    In particular, we will consider $\Delta f$ that is symmetric, that is, for any 
    $T$ and $T'$ with $|T|=|T'| = i$, we have $\Delta f(T) = \Delta f(T') = \wh{f}(i)$
    for some $\wh{f}:\m\mapsto {\mathbb R}$. Note that for symmetric functions, the 
    $W$-representation is also symmetric, that is, given by $\Delta w(S) = \Delta w(S') = \wh{w}(j)$
    whenever $|S| = |S'| = j$. One can easily get a relation between $\wh{w}$ and $\wh{f}$ as follows:
    \begin{align}
       \wh{w}(j) & = \Delta w(S) = \sum_{T : S \cup T = [m]}
        (-1)^{|S \cap T| + 1} \Delta f(T) &  \notag \\
        & = \sum_{i = 0}^m \sum_{T : S
        \cup T = [m], |T| = i} (-1)^{i+j+m+1} \wh{f}(i)\notag  \\
        & = \sum_{i = 0}^m {|S| \choose i-(m-|S|)} (-1)^{i+j+m+1}
        \wh{f}(i) \notag \\
        & = \sum_{i = 0}^m {j \choose m-i} (-1)^{j+(m-i)+1} \wh{ f}(i)
        \notag \\
        & =  \sum_{i = 0}^m {j \choose i} (-1)^{i+j+1}\wh{f}(m-i)
        \label{eq:wf}
    \end{align}
\noindent
In the first equality $S$ is an arbitrary subset of size $j$. We now show that there exists a choice of $\wh{f}:\m\mapsto {\mathbb R}$ such that (a') $\wh{f}(i) = 0$ for $3m/8 \le i \le 5m/8$. Note that this will imply $\Delta f$ is zero on at least 
$(1 - e^{-\Theta(m)})$ subsets implying condition (a). Furthermore, the choice of $\wh{f}$ will imply that (b') $\wh(j) \ge 1$ whenever $m > j > m/4$, and $\wh(j) \ge -N$ otherwise. This implies condition (b). 

For this, let $\alpha_i := (-1)^{i+1}\wh{f}(m-i)$. From \eqref{eq:wf}, we get $(-1)^j\wh{w}(j) = \sum_{i=0}^m \alpha_i{j\choose i}$. We consider the RHS as a polynomial over $j$, and in fact, the degree $i$ polynomials ${j \choose i} := \frac{j(j-1)\ldots(j-i+1)}{i!}$ form what is known as the Mahler bases of rational polynomials (see, for example, \cite{Mahler,Robert}). 

As a result, one can choose $\alpha_i$ for $0\le i < 3m/8$ such that $\sum_{i=0}^{3m/8} \alpha_i{j \choose i}$ is {\em any} desired rational polynomial of degree $(3m/8 - 1)$. In particular, we choose $\alpha_i$'s so that 
\begin{equation} 
    \sum_{i=0}^{3m/8-1} \alpha_i {j \choose i} = h_1(j) = 4 (-1)^{5m/8}
    \prod_{k=m/4+1}^{5m/8-1} (j-k-1/2) \label{eq:alpha1}
\end{equation}
Similarly,  $\sum_{m>i>5m/8}\alpha_i{j \choose i}$ can be chosen to be $j(j-1)\cdots (j - 5m/8)g(j)$ for any degree $(3m/8 - 2)$ polynomial. We choose $\alpha_i$'s for $5m/8 < i < m$ so that 
\begin{equation}
    \sum_{5m/8 < i < m} \alpha_i {j \choose i} = h_2(j) = (20(m!)+4)
    (-1)^{m-1} j(j-1)\dots(j-5m/8) \prod_{k=5m/8+1}^{m-2} (j-k-1/2)
    \label{eq:alpha2}
\end{equation}
Finally, as promised, we let $\alpha_i = 0$ for $3m/8 \le i \le 5m/8$. We now argue that 
condition (b') holds. Note that $\wh{w}(j) = (-1)^j(h_1(j) + h_2(j))$. 
\bigskip

\noindent
\underline{If $m/4 < j \le 5m/8$}: From \eqref{eq:alpha2}, $h_2(j) = 0$. Also, from \eqref{eq:alpha1}, we get that 
the sign of $h_1(j)$ for $m/4 < j < 5m/8$ is precisely $(-1)^{5m/8}(-1)^{5m/8 - j} = (-1)^j$. So, $(-1)^jh_1(j)$ is 
positive. Furthermore, the absolute value of $h_1(j)$ is at least $1$, implying $\wh{w}(j) \ge 1$. 
\bigskip

\noindent
\underline{If $5m/8 < j < m$}: We use that $\wh{w}(j) \ge (-1)^jh_2(j) - |h_1(j)|$. The former term is at least $5m!$ 
via a similar reasoning as above. $|h_1(j)|$, as follows from \eqref{eq:alpha1}, is at most $4m! $. This is because each term in the product is at most $m!$ in absolute value. This gives $\wh{w}(j) \ge m! \ge 1$ in this range. 
\bigskip

\noindent
\underline{If $0 \le j \le m/4$, or $j=m$}: Once again, we get that $\wh{w}(j) = (-1)^jh_1(j)$ which changes it's sign
as $j$ changes. However, the absolute value is at most $4m!$, so choosing $N=5m!$, we get $\wh{w}(j)\ge -N$. 

\noindent
Thus, condition (b') is also satisfied, in turn implying that (b) is satisfied. 
\end{proof}

\begin{proof}[\pref{lem:WnearUfar}]
   Consider the function $f$ whose $W$-representation satisfies that
    $w(S) = m$ if $|S| = 1$, $w(S) = -1$ if $|S| = 2$, and $w(S) = 0$ if
    $|S| \ge 3$. 
        
    We first note that for any subset $T$ and any $i, j \notin T$,
    we have $f(T+i+j) - f(T+i) - f(T+j) + f(T) = - \sum_{S : i, j \in S}
    w(S) = -w({i, j}) = 1$. Therefore, the function is supermodular. So 
    for any subset $R$ and any $i \notin R$, we have $f(R+i) - f(R) \ge 
    f(i) - f(\emptyset) = \sum_{S : i \in S} w(S) = m - (m-1) = 1$.
    Hence, the function is monotonely increasing. Note that $f(\emptyset)
    = 0$. We get that $f$ is non-negative.

    Next, we will show that $f$ is at least $1/4$-far from coverage
    functions. Let us partition all the $2^m$ subsets into groups of
    size $4$ such that for any subset $S$ of $[m]-i-j$, we let $S$,
    $S+i$, $S+j$, and $S+i+j$ be in the same group. Note that the
    function is strictly supermodular yet any coverage function must be
    submodular. So at least one of the four function values in each
    group need to be changed in each group in order to make it a
    coverage function. \end{proof}